\documentclass[pra,aps,twocolumn,superscriptaddress]{revtex4-1}
\usepackage[dvips]{graphicx}
\usepackage{amsmath,amssymb,amsthm,dsfont,fullpage,hyperref}

\newcommand{\bra}[1]{\langle#1|}
\newcommand{\ket}[1]{|#1\rangle}

\newcommand{\Tr}{\operatorname{Tr}}

\newcommand{\hilb}[1]{\mathcal{#1}}
\newcommand{\set}[1]{\mathcal{#1}}

\newcommand{\supp}{\operatorname{supp}}

\newtheorem{dfn}{Definition}
\newtheorem{lmm}{Lemma}
\newtheorem{thm}{Theorem}
\newtheorem{cor}{Corollary}

\begin{document}

\title{Tight Bounds on Accessible Information and Informational Power}

\author{Michele \surname{Dall'Arno}}

\affiliation{Graduate School of Information Science, Nagoya
  University, Chikusa-ku, Nagoya, 464-8601, Japan}

\author{Francesco \surname{Buscemi}}

\affiliation{Graduate School of Information Science, Nagoya
  University, Chikusa-ku, Nagoya, 464-8601, Japan}

\affiliation{Institute for Advanced Research, Nagoya
  University, Chikusa-ku, Nagoya 464-8601, Japan}

\author{Masanao \surname{Ozawa}}

\affiliation{Graduate School of Information Science, Nagoya
  University, Chikusa-ku, Nagoya, 464-8601, Japan}

\date{\today}

\begin{abstract}
  The accessible information quantifies the amount of
  classical information that can be extracted from an
  ensemble of quantum states. Analogously, the informational
  power quantifies the amount of classical information that
  can be extracted by a quantum measurement. For both
  quantities, we provide upper and lower bounds that depend
  only on the dimension of the system, and we prove their
  tightness. In the case of symmetric informationally
  complete (SIC) ensembles and measurements, stronger bounds
  are provided and their tightness proved for qubits and
  qutrits. From our upper bounds, we notice, perhaps
  surprisingly, that the statistics generated by SIC
  ensembles or measurements in arbitrary dimension, though
  optimal for tomographic purposes, in fact never contain
  more than just one bit of information, the rest being
  constituted by completely random bits. On the other hand,
  from our lower bounds, we obtain an explicit strategy
  beating the so-called ``pretty-good'' one for the
  extraction of mutual information in the case of SIC
  ensembles and measurements.
\end{abstract}

\maketitle

\section{Introduction}

Quantum theory poses bounds on the amount of classical
information that can be extracted from an ensemble of
quantum states or by a quantum measurement, and the problem
of quantifying such bounds lies at the heart of many
important aspects of the theory. From a fundamental
viewpoint, it has implications in fields such as quantum
measurement theory~\cite{Oza13}, entropic uncertainty
relations~\cite{BW07, BR10, WW10, CCYZ12, BHOW13}, or the
informational axiomatization of quantum
theory~\cite{CDP11,CDP12}. On the other hand, in
applications it is often important to implement preparations
and measurements that achieve those bounds in order to
optimally perform tasks such as information
locking~\cite{WW10}, private quantum
decoupling~\cite{BGK09,Bus09}, purification of noisy quantum
measurements~\cite{DDS10}, randomness
generation~\cite{NC00}, classical communication over quantum
channels~\cite{NC00}, storage and retrieval of information
from quantum memories~\cite{NC00}, and quantum error
correction~\cite{Bus08,CDDMP10}.

The problem of quantifying how much information can be
extracted from an ensemble of quantum states was first
addressed almost half a century ago, and is usually referred
to as the accessible information problem~\cite{LL66, Hol73,
  Bel75a, Bel75b, Dav78, Oza84, Oza86, YO93, JRW94}. Since
exact solutions to this problem are known only for some
restricted classes of symmetric preparations, tight upper
and lower bounds - known as the Holevo-Yuen-Ozawa
bound~\cite{Hol73, YO93} and the subentropy
bound~\cite{JRW94}, respectively - can often be extremely
useful. We note however that, even for highly symmetric
preparations like the so-called {\em tetrahedral
  preparation} first introduced by Davies~\cite{Dav78}, the
analytic evaluation of the accessible information remained
an open problem for long time~\cite{DDS11}.

Very recently, a ``dual'' problem to that of accessible
information has been introduced, as the problem of
quantifying how much information can be extracted by a
quantum measurement (i.e. its {\em informational
  power})~\cite{DDS11}. Despite the usefulness of such dual
approach~\cite{DDS11, OCMB11, Hol12, Hol13, SS14, Szy14}
(leading e.g. to the analytic evaluation~\cite{DDS11,
  OCMB11, SS14} of the accessible information of the
tetrahedral preparation), tight upper and lower bounds on
the informational power of a quantum measurement are still
lacking.

The tetrahedral configuration for qubits is generalized to
higher dimensions by what are called {\em symmetric
  informationally complete} (SIC) preparations and
measurements~\cite{Cav99}. SIC quantum
measurements~\cite{Cav99, RBSC04, AFF11} were introduced as
useful tools for quantum tomography, quantum
cryptography~\cite{FS03, REK04, Ren05, ZB05, Bal05, Kim07},
and classical signal processing~\cite{Sco06}, but it has
been argued that they also play a fundamental role in
characterizing the structure of quantum state
space~\cite{AFF11}. Despite numerous efforts, many important
properties of SIC measurements, not last their existence in
any dimension, are still unanswered problems. Therefore, any
approach that can shed new light on SIC measurements may
lead to a better understanding of their structure.

The purpose of this work is twofold. First, we provide tight
upper and lower bounds, depending only on the dimension of
the system, on the accessible information and the
informational power of any preparation and
measurement. While both upper and lower bounds for
accessible information can be trivially saturated, we find
that, for the case of informational power, the lower bound
is saturated by a non-trivial measurement that we name
``Scrooge'' following Ref.~\cite{JRW94}. Second, we focus on
the case of SIC preparations and measurements, providing
stronger upper and lower bounds on accessible information
and informational power, and proving their tightness in the
case of qubits and qutrits. As a corollary of our upper
bound, we show that the accessible information and the
informational power of SIC preparations and measurements are
always upper bounded by one bit. In this sense, therefore,
the statistics generated by SIC preparations or SIC
measurements, though optimal for tomographic purposes, are
(perhaps surprisingly) constituted by almost completely
random bits. As a corollary of our lower bound, we provide a
strategy for the extraction of mutual information from SIC
ensembles or by SIC measurements, which beats the so-called
``pretty-good'' one~\cite{HW94, Hal97, Bus07, BH09}.

The paper is organized as follows. In
Sec.~\ref{sec:povmbound} we provide our first main result,
namely tight upper and lower bounds on the accessible
information and informational power for arbitrary
preparations and measurements, as a function of the
dimension only. In Sec.~\ref{sec:sicbound} we provide our
second main result, namely upper and lower bounds on the
accessible information and informational power for SIC
preparations and measurements. For qubits and qutrits we prove
tightness in Sec.~\ref{sec:sicqubit}. We conclude the paper
discussing some open problems in Sec.~\ref{sec:conclusion}.

\section{Bounds for arbitrary ensembles and measurements}
\label{sec:povmbound}

In this Section we provide upper and lower bounds on the
accessible information (resp., informational power) of
arbitrary preparations (resp., measurements).

Let us recall some basic definitions~\cite{Cov06} from
classical information theory. A random variable $X$ is a
function that maps from its domain, the sample space, to its
range, the real numbers, according to a probability
distribution (probability density for the continuous case)
$p(X{=}x)$. Given a random variable $X$, the Shannon entropy
$H(X)$ defined as
\begin{align*}
  H(X) := - \sum_x p(X{=}x) \log p(X{=}x),
\end{align*}
for the discrete case, or
\begin{align*}
  H(X) := - \int_{-\infty}^{\infty} dx \ p(X{=}x) \log
  p(X{=}x),
\end{align*}
for the continuous case, is a measure of the lack
of information about the outcome of $X$. We write $\log$ for
binary logarithms and $\ln$ for natural logarithms, and we
express informational quantities in bits. Given two random
variables $X$ and $Y$, the joint Shannon entropy $H(X,Y)$ is
defined as
\begin{align*}
  H(X,Y) := - \sum_{x,y} p(X{=}x, Y{=}y) \log p(X{=} x,
  Y{=}y),
\end{align*}
for the discrete case, or
\begin{align*}
  & H(X,Y) := \\ & - \iint_{-\infty}^{\infty} dx dy \ p(X{=}x,
  Y{=}y) \log p(X{=}x, Y{=}y)
\end{align*}
for the continuous case, and is a measure of the lack of
information about the joint outcomes of $X$ and $Y$. The
conditional Shannon entropy $H(X|Y)$ defined as $H(Y|X) :=
H(X,Y) - H(X)$ is a measure of the lack of information about
the outcome of $X$ given the knowledge of the outcome of
$Y$. The {\em mutual information} $I(X;Y)$ defined as
$I(X;Y) := H(X) + H(Y) - H(X,Y) = H(Y) - H(Y|X) = H(X) -
H(X|Y)$ is a measure of correlation between $X$ and $Y$, or,
equivalently, a measure of how much information about each
random variable is carried by the other one.

Let us recall some basic definitions~\cite{NC00} from
quantum information theory. Any quantum system is associated
to an Hilbert space $\hilb{H}$, and we denote with
$L(\hilb{H})$ the space of linear operators on
$\hilb{H}$. We will consider only finite dimensional Hilbert
spaces. A {\em state} $\rho$ is a positive semidefinite
operator in $L(\hilb{H})$ such that $\Tr[\rho] \le 1$. Any
preparation of a quantum system is described by an {\em
  ensemble}. In the discrete (resp., continuous) case, an
ensemble is an operator valued measurable function $\set{E}
= \{ \rho_x \}$ from real numbers $x$ to states $\rho_x \in
L(\hilb{H})$, associated with a probability distribution
(resp., probability density), such that $\sum_x\Tr[\rho_x] =
1$ (resp., $\int_x dx \Tr[\rho_x] = 1$). We will call any
ensemble of rank-one states a {\em rank-one ensemble}. For
any Hilbert space $\hilb{H}$, the continuous ensemble
uniformly distributed according to Haar measure~\cite{Ste95,
  Syk74} is sometimes called {\em Scrooge
  ensemble}~\cite{JRW94}.

An {\em effect} $\Pi$ is a positive semidefinite operator in
$L(\hilb{H})$ such that $\Pi \le \openone$. Any measurement
on a quantum system is described by a {\em probability
  operator-valued measure} (POVM)~\cite{Hol73}. In the
discrete (resp., continuous) case, a POVM is an operator
valued measurable function $\set{P} = \{ \Pi_y \}$ from real
numbers $y$ to effects $\Pi_y \in L(\hilb{H})$, such that
$\sum_y \Pi_y = \openone$ (resp., $\int_y dy \Pi_y =
\openone$), where $\openone$ denotes the identity operator.
We will call any POVM of rank-one effects a {\em rank-one
  POVM}. In analogy with~\cite{JRW94}, for any Hilbert space
$\hilb{H}$, we call the continuous POVM uniformly
distributed according to Haar measure~\cite{Ste95, Syk74}
the {\em Scrooge POVM}.

Given an ensemble $\set{E} = \{ \rho_x \}$ and a POVM
$\set{P} = \{ \Pi_y \}$, the joint probability $P(X{=}x,
Y{=}y)$ of state $\rho_x$ and outcome $\Pi_y$ is given by
the Born rule, namely $P(X{=}x, Y{=}y) = \Tr[\rho_x \Pi_y]$,
and the mutual information $I(X; Y)$ is usually denoted with
$I(\set{E}, \set{P})$. Given an Hilbert space $\hilb{K}$ and
an ensemble $\set{E} = \{ \rho_x \}$ on $\hilb{K}$ with
average state $\rho := \sum_x \rho_x$, or $\rho := \int_x dx
\rho_x$ in the continuous case, call $P$ the orthogonal
projector on $\hilb{H} := \supp\rho$. Then for any POVM
$\set{P} = \{ \Pi_y \}$ on $\hilb{K}$ one has that
\begin{align*}
  \Tr[\rho_x \Pi_y] = \Tr[\rho_x P \Pi_y P], \qquad \forall
  x,y,
\end{align*}
since $\rho_x = P \rho_x P$ for any $x$. Since $\int_y dy P
\Pi_y P = P$ (i. e. $\sum_y P \Pi_y P = P$ in the discrete
case) and $P \Pi_y P \ge 0$ for any $y$, one has that $\{ P
\Pi_y P\}$ is a POVM on $\hilb{H}$ with same joint
probability. Therefore we can limit ourselves, without loss
of generality, to ensembles $\set{E} = \{ \rho_x \}$ such
that the average state $\rho$ is {\em invertible} (or {\em
  faithful}~\cite{Oza85}), namely to be defined on
$\hilb{H}$ with $\hilb{H} := \supp\rho$. Given an ensemble
$\set{E} = \{ \rho_x \}$ with average state $\rho$, then
$\{\rho^{-1/2} \rho_x \rho^{-1/2}\}$ is a POVM, sometimes
called pretty-good POVM~\cite{HW94, Hal97, Bus07, BH09} for
$\set{E}$ or $\rho^{-1}$-distortion~\cite{JRW94}. Given a
POVM $\set{P} = \{ \Pi_y \}$ and a normalized state $\rho$,
then $\{\rho^{1/2} \Pi_y \rho^{1/2}\}$ is an ensemble with
average state $\rho$, sometimes called pretty-good ensemble
for $\set{E}$ or $\rho$-distortion.

The accessible information~\cite{LL66} was introduced as a
measure of how much information can be extracted from an
ensemble.
\begin{dfn}[Accessible Information]\label{def:accinfo}
  The accessible information $A(\set{E})$ of a discrete or
  continuous ensemble $\set{E}$ is the supremum over any
  discrete or continuous POVM $\set{P}$ of the mutual
  information $I(\set{E}, \set{P})$, namely
  \begin{align*}
    A(\set{E}) := \sup_{\set{P}} I(\set{E}, \set{P}).
  \end{align*}
  Any POVM $\set{P}$ that attains the accessible information
  is called maximally informative for the ensemble
  $\set{E}$.
\end{dfn}

The informational power~\cite{DDS11} was introduced as a
measure of how much information can be extracted by a POVM.
\begin{dfn}[Informational Power]\label{def:infopower}
  The informational power $W(\set{P})$ of a discrete
    or continuous POVM $\set{P}$ is the supremum over any
  discrete or continuous ensemble $\set{E}$ of the
  mutual information $I(\set{E}, \set{P})$, namely
  \begin{align*}
    W(\set{P}) := \sup_{\set{E}} I(\set{E}, \set{P}).
  \end{align*}
  Any ensemble $\set{E}$ that attains the informational
  power is called maximally informative for the POVM
  $\set{P}$.
\end{dfn}

The following Lemma, appearing also in Ref.~\cite{DDS11},
provides a fundamental relation between accessible
information and informational power.
\begin{lmm}
  \label{lmm:duality}
  For any discrete or continuous POVM $\set{P} = \{ \Pi_y
  \}$, the informational power $W(\set{P})$ is the supremum
  over normalized states $\rho$ of the accessible
  information of the ensemble $\{ \rho^{1/2} \Pi_y
  \rho^{1/2} \}$, namely
  \begin{align}
     W(\set{P}) = \sup_{\rho} A( \{ \rho^{1/2} \Pi_y
     \rho^{1/2} \}). \label{eq:duality}
  \end{align}
\end{lmm}

\begin{proof}
  For both the discrete or continuous case, one has that
  \begin{align*}
    W(\set{P}) & = \sup_{\set{E}} I(\set{E},\set{P}) \\ & =
    \sup_{\set{E}} I(\{ \rho^{-1/2} \rho_x \rho^{-1/2} \},
    \{ \rho^{1/2} \Pi_y \rho^{1/2} \}) \\ & = \sup_{\rho} A(
    \{ \rho^{1/2} \Pi_y \rho^{1/2} \}),
  \end{align*}
  where first equality follows from
  Definition~\ref{def:infopower}; second equality follows
  from the identity $\Tr[ \rho_x \Pi_y ] = \Tr[
    \rho^{-1/2} \rho_x \rho^{-1/2} \rho^{1/2}
    \Pi_y \rho^{1/2}]$; third equality follows from
  Definition~\ref{def:accinfo}.
\end{proof}

We can now provide the first main result of this work.
\begin{thm}[Bounds for arbitrary ensembles and POVMs]
  \label{thm:povmbound}
  For any $d$-dimensional discrete or continuous rank-one
  ensemble $\set{E} = \{ \rho_x \}$ with average state
  $\rho$ and any $d$-dimensional discrete or continuous
  rank-one POVM $\set{P} = \{ \Pi_y \}$ one has that
  \begin{align}
    0 \le A(\set{E}) \le
    \log(d),\label{eq:accinfobound}
  \end{align}
  and
  \begin{align}
    \log(d) - \frac1{\ln(2)} \sum_{n=2}^d \frac1n \le &
    W(\set{P}) \le \log(d)\label{eq:infopowerbound},
  \end{align}
  Lower and upper bounds in Eq.~\eqref{eq:accinfobound} are
  saturated, respectively, when $\rho_x \propto \rho$ for
  all $x$, and when $\rho_x = \frac1d \ket{e_x}\bra{e_x}$
  for some orthonormal basis $\ket{e_x}$. Lower and upper
  bounds in Eq.~\eqref{eq:infopowerbound} are saturated,
  respectively, when $\set{P}$ is the Scrooge POVM and when
  $\Pi_y = \ket{e_y}\bra{e_y}$ for some orthonormal basis
  $\ket{e_y}$.
\end{thm}

\begin{proof}
  Lower bound in Eq.~\eqref{eq:accinfobound} follows from
  the non-negativity of mutual information.
  
  Upper bound in Eq.~\eqref{eq:accinfobound} is the well
  known Holevo-Yuen-Ozawa bound~\cite{Hol73, YO93}.

  Let us prove the lower bound in
  Eq.~\eqref{eq:infopowerbound}. One has that
  \begin{align*}
    \inf_{\set{P} | \set{P} \textrm{ is rank-$1$}} W(\set{P}) & =
    \inf_{\set{P} | \set{P} \textrm{ is rank-$1$}} \sup_{\rho} A(
    \{ \rho^{1/2} \Pi_y \rho^{1/2} \}) \\ & =
    \sup_{\rho} A( \{ \rho^{1/2} \Pi_y^*
    \rho^{1/2} \}) \\ & = \log(d) - \frac1{\ln(2)}
    \sum_{n=2}^d \frac1n,
  \end{align*}
  where first equality follows from
  Lemma~\eqref{lmm:duality}; second equality follows from
  the fact proved in Sec. II and Sec. III of
  Ref.~\cite{JRW94} that the infimum over $\set{P}$ for any
  $\rho$ is attained when $\set{P}$ is the Scrooge POVM
  $\set{P}^* = \{ \Pi_y^* \}$, and thus infimum and supremum
  commute; third equality follows from the fact proved in
  Sec. IV of Ref.~\cite{JRW94} that the supremum over
  $\rho$, in the case of Scrooge POVM, is attained when
  $\rho = \frac1d \openone$. The lower bound in
  Eq.~\eqref{eq:infopowerbound} is then proved.

  Let us prove the upper bound in
  Eq.~\eqref{eq:infopowerbound}. By absurd suppose that
  there exists a POVM $\set{P}^*$ such that $W(\set{P}^*) >
  \log(d)$, and call $\set{E}^*$ its maximally informative
  ensemble. Then $A(\set{E}^*) \ge W(\set{P}^*) > \log(d)$
  contradicting Holevo-Yuen-Ozawa bound. The upper bound in
  Eq.~\eqref{eq:infopowerbound} is then proved.
\end{proof}

Notice that, if the hypotheses on the rank of $\set{E}$ and
$\set{P}$ are dropped in Theorem~\ref{thm:povmbound}, while
Eq.~\ref{eq:accinfobound} still holds, the lower bound in
Eq.~\eqref{eq:infopowerbound} does not hold anymore and must
be replaced with the bound $0 \le W(\set{P})$, which is
trivially saturated when $\Pi_y \propto \openone$ for all
$y$. Finally, notice that the expression in the left hand
side of Eq.~\eqref{eq:infopowerbound} appeared in Eq.~(42)
of a very recent work~\cite{SS14} by S{\l}omczy\'{n}ski and
Szymusiak as the average value of the {\em relative entropy
  of measurement} over all pure states.

\section{Bounds for SIC ensembles and SIC measurements}
\label{sec:sicbound}

In this Section we provide upper and lower bounds on the
accessible information (resp., informational power) of
rank-one SIC ensembles (resp., rank-one SIC POVMs).

Let us first introduce symmetric informationally complete
(SIC) sets of operators.
\begin{dfn}[SIC set of operators]
  \label{dfn:sicset}
  A $d$-dimensional set $X = \{ X_x \}_{x=1}^{d^2}$ of $d^2$
  rank-one positive operators $X_x$ satisfying $\Tr[X_x] =
  \lambda$ for any $x$ and for some $\lambda$ and
  $\Tr[X_xX_y] = \frac{\lambda^2}{d+1}$ for any $x \ne y$ is
  called symmetric informationally complete (SIC) set of
  operators.
\end{dfn}

The following Lemma, appearing also in Ref.~\cite{AFF11},
provides a fundamental property of SIC set of operators.
\begin{lmm}\label{lmm:average}
  For any $d$-dimensional SIC set $X = \{ X_x \}$ of
  operators $X_x$ with $\Tr[ X_x ] = \lambda$, the average
  operator $\bar X := \sum_x X_x$ is $\bar X = d \lambda
  \openone$.
\end{lmm}

\begin{proof}
  For any operator $A$ we have that $\Tr[A^\dagger A] = 0$
  if and only if $A=0$. Taking $A = d \lambda \openone -
  \bar X$ we have
  \begin{align*}
    & \Tr[(d \lambda \openone - \bar X)^2] \\ = & d^2
    \lambda^2 \Tr[\openone] + \Tr[\bar X^2] - 2 d \lambda
    \Tr[\bar X] \\ = & d^3 \lambda^2 + d^2 \lambda^2 +
    \frac{d^2(d^2-1)}{d+1} \lambda^2 -2 d^3 \lambda^2 \\ = &
    0.
  \end{align*}
  by noticing that $\Tr[\bar X] = \sum_{x=1}^{d^2} \lambda =
  d^2 \lambda$ and $\Tr[\bar X^2] = \sum_x \lambda^2 +
  \sum_{x \ne y} \frac{\lambda^2}{d+1} = d^2 \lambda^2 +
  \frac{d^2(d^2-1)}{d+1} \lambda^2$.
\end{proof}

\begin{dfn}[SIC Ensemble]\label{dfn:sicensamble}
  A $d$-dimensional SIC set $\set{E} = \{ \rho_x \}$ of
  operators $\rho_x$ with $\Tr[ \rho_x ] = \frac1{d^2}$ is
  called symmetric informationally complete (SIC) ensemble.
\end{dfn}

\begin{dfn}[SIC POVM]\label{dfn:sicpovm}
  A $d$-dimensional SIC set $\set{P} = \{ \Pi_y \}$ of operators
  $\Pi_y$ with $\Tr[ \Pi_y ] = \frac1d$ is called symmetric
  informationally complete (SIC) POVM.
\end{dfn}

Notice that Lemma~\ref{lmm:average} ensures that the average
state of any SIC ensemble is $\frac1d \openone$ and that in
Definition~\ref{dfn:sicpovm} the necessary condition $\sum_y
\Pi_y = \openone$ is automatically granted. Notice also that
given any SIC POVM $\set{P} = \{ \Pi_x \}$, one has that $\{
\frac1d \Pi_x \}$ is a SIC ensemble. Analogously, given any
SIC ensemble $\set{E} = \{ \rho_x\}$, one has that $\{ d
\rho_x \}$ is a SIC POVM, since the average state of any SIC
ensemble is $\frac1d \openone$. Then there is a one-to-one
correspondence between SIC POVMs and SIC ensembles given by
renormalization.

We can now provide the second main result of this work.
\begin{thm}[Bounds for SIC ensembles and SIC POVMs]
  \label{thm:sicbound}
  For any $d$-dimensional SIC ensemble $\set{E} = \{ \rho_x
  \}$ and any $d$-dimensional SIC POVM $\set{P} = \{ \Pi_y \}$
  one has that
  \begin{align}
    \log(d) - \frac1{\ln 2} \sum_{n=2}^d \frac1n \le &
    A(\set{E}) \le
    \log\frac{2d}{d+1}, \label{eq:sicaccinfobound}
  \end{align}
  and
  \begin{align}
    \log(d) - \frac1{\ln 2} \sum_{n=2}^d \frac1n \le &
    W(\set{P}) \le
    \log\frac{2d}{d+1}\label{eq:sicinfopowerbound}.
  \end{align}
\end{thm}

\begin{proof}
  The lower bound in Eq.~\eqref{eq:sicinfopowerbound}
  follows from Eq.~\eqref{eq:infopowerbound}.
  
  Let us prove the upper bound in
  Eq.~\eqref{eq:sicinfopowerbound}. Upon introducing an
  ensemble $\set{E} = \{ \rho_x \}$ and two random variables
  $X$ and $Y$ such that $p(X{=}x) = \Tr[\rho_x]$ and
  $p(X{=}x,Y{=}y) = \Tr[\rho_x\Pi_y]$, we have
  $I(\set{E},\set{P}) = I(X;Y) = H(Y) - H(Y|X)$. Entropy
  $H(Y)$ is trivially upper bounded by $H(Y) \le
  \log(d^2)$. Since the entropy $H(Y|X{=}x)$ is lower
  bounded by $H(Y|X{=}x) \ge \log(d(d+1)/2)$ as proven in a
  recent interesting work by Rastegin~\cite{Ras13}, the
  conditional entropy $H(Y|X) = \sum_x p(X{=}x) H(Y|X{=}x)$
  is lower bounded by $H(Y|X) \ge \log(d(d+1)/2)$. The upper
  bound in Eq.~\eqref{eq:sicinfopowerbound} is then proved.

  The upper bound in Eq.~\eqref{eq:sicaccinfobound} follows
  from the upper bound in Eq.~\eqref{eq:sicinfopowerbound}
  and from Lemma~\ref{lmm:duality}.

  Let us prove the lower bound in
  Eq.~\eqref{eq:sicaccinfobound}. One has that
  \begin{align*}
    \inf_{\set{E} | \set{E} \textrm{ is SIC}} A(\set{E}) &
    \ge \inf_{\set{P} | \set{P} \textrm{ is rank-$1$}}
    A\left(\frac1d \set{P}\right)\\ & = \log(d) -
    \frac1{\ln(2)} \sum_{n=2}^d \frac1n.
  \end{align*}
  where first inequality follows from the fact that the set
  of SIC ensembles is a subset of the set of rank-one
  ensembles with average state $\frac1d \openone$ due to
  Lemma~\ref{lmm:average}; second equality follows from the
  fact proved in Sec. II and III of Ref.~\cite{JRW94} that
  the infimum over $\set{P}$ is attained when $\set{P}$ is
  the Scrooge POVM. The lower bound in
  Eq.~\eqref{eq:infopowerbound} is then proved.
\end{proof}

Notice that, for any dimension $d$, the lower bound on
accessible information and informational power given by
Eqs.~\eqref{eq:sicaccinfobound}
and~\eqref{eq:sicinfopowerbound} are tighter than those
provided by the pretty-good POVM and ensemble. Indeed, for
any $d$-dimensional SIC ensemble $\set{E} = \{ \rho_x \}$,
the corresponding pretty-good POVM is $\{ d \rho_x \}$;
analogously, for any $d$-dimensional SIC POVM $\set{P} = \{
\Pi_x \}$, the corresponding pretty-good ensemble is $\{
d^{-1} \Pi_x\}$; the mutual information given by the
pretty-good strategy is then
\begin{align*}
  & I(\{ \rho_x \}, \{ d \rho_x \})\\ = & I(\{ d^{-1} \Pi_x \},
  \{ \Pi_x \}) \\ = & \frac{2d}{d^2(d+1)}\log d -
  \frac{d-1}{d^2(d+1)} \log (d+1) \\ \le & \log d -
  \frac1{\ln 2} \sum_{n=2}^d \frac1n,
\end{align*}
where the first and second equalities follow from
Definition~\ref{dfn:sicset} and the inequality follows by
direct inspection.

Notice also that the values of the mutual information in
lower bounds in Eqs.~\eqref{eq:sicaccinfobound}
and~\eqref{eq:sicinfopowerbound} can be achieved by the
Scrooge measurement and ensemble, respectively, as shown in
the proof of Theorem~\ref{thm:sicbound}. Therefore, when
extracting information from a SIC ensemble or by a SIC
measurement, the strategy based on the Scrooge construction
performs strictly better than the pretty-good one. Note,
however, that pretty-good construction was not introduced
with respect to accessible information extraction, but
aiming to maximize the probability of correct guess.

Finally, as a ``numerological'' aside, we notice the
intriguing fact that the quantity $\frac{2}{d+1}$, appearing
in the right hand side of both
Eqs.~\eqref{eq:sicaccinfobound} and
~\eqref{eq:sicinfopowerbound}, appears also in a number of
other quantities related to SIC POVMs. For example, the {\em
  index of coincidence} $C$, defined for any POVM $\set{P} =
\{ \Pi_y \}$ and any state $\rho$ as $C = \sum_y (\Tr[\rho
  \Pi_y])^2$, is constant and equal to $\frac{2}{d(d+1)}$ in
the case of SIC POVMs and pure
states~\cite{KR05,Ras13}. Also, the so-called {\em
  quantumness}~\cite{FS03} $Q$ of a Hilbert space,
representing the worst-case difficulty of transmitting
quantum states through a purely classical communication
channel, was shown~\cite{Fuc04} to be achieved whenever the
states to be transmitted constitute a SIC ensemble, in which
case $Q = \frac{2}{d+1}$. As a further example, the {\em
  total uncertainty}~\cite{ADF14} $T$ on the outcome of a
complete set of mutually unbiased bases~\cite{Ivo81} in any
prime dimension is known to be lower bounded as $T \ge (d+1)
\log\frac{d+1}{2}$, with equality if and only if the input
states constitute a SIC ensemble~\cite{ADF14}.

\subsection{The cases of qubits and qutrits}
\label{sec:sicqubit}

The following two corollaries prove the tightness of the
upper bounds in Eqs.~\eqref{eq:sicaccinfobound} and
\eqref{eq:sicinfopowerbound} for qubits and qutrits.
\begin{cor}
  \label{thm:sicqubit}
  Upper bounds in Eqs.~\eqref{eq:sicaccinfobound} and
  \eqref{eq:sicinfopowerbound} are tight in dimension $d=2$,
  namely the $2$-dimensional SIC ensemble (tetrahedral
  ensemble) $\set{E}$ and the $2$-dimensional SIC POVM
  (tetrahedral POVM) $\set{P}$ are such that
  \begin{align*}
    A(\set{E}) = W(\set{P}) = \log\frac43.
  \end{align*}
  The accessible information is achieved by the so-called
  antitetrahedral POVM and the informational power by the
  so-called antitetrahedral ensemble.
\end{cor}

\begin{proof}
  The POVM $\set{P} = \{ \Pi_y \}_{y=1}^4$ (tetrahedral POVM)
  with $\Pi_y = \frac12 \ket{\pi_y}\bra{\pi_y}$ given by
  \begin{align*}
    \ket{\pi_1} &= \ket{0},\\ \ket{\pi_2} &=
    \frac{1}{\sqrt3} \ket{0} + \sqrt{\frac23}
    \ket{1},\\ \ket{\pi_3} &= \frac{1}{\sqrt3} \ket{0} +
    e^{i\frac23\pi} \sqrt{\frac23} \ket{1},\\ \ket{\pi_4} &=
    \frac{1}{\sqrt3} \ket{0} + e^{-i\frac23\pi}
    \sqrt{\frac23} \ket{1}.
  \end{align*}
  is a SIC POVM in dimension $2$.

  The ensemble $\set{E} = \{ \rho_x \}_{x=1}^4$
  (antitetrahedral ensemble) with $\rho_x = \frac14
  \ket{\psi_x}\bra{\psi_x}$ given by
  \begin{align*}
    \ket{\psi_1} &= \ket{1},\\ \ket{\psi_2} &=
    \sqrt{\frac23} \ket{0} - \frac{1}{\sqrt3}
    \ket{1},\\ \ket{\psi_3} &= \sqrt{\frac23} \ket{0} +
    e^{i\frac13\pi}\frac{1}{\sqrt3}
    \ket{1},\\ \ket{\psi_4} &= \sqrt{\frac23} \ket{0} +
    e^{-i\frac13\pi} \frac{1}{\sqrt{3}} \ket{1}.
  \end{align*}
  saturates the upper bound in
  Eq.~\eqref{eq:sicinfopowerbound}. Since the average state
  of the ensemble $\set{E}$ is $\frac1d \openone$, due to
  Lemmas~\ref{lmm:duality} and~\ref{lmm:average} also the
  upper bound in Eq.~\eqref{eq:sicaccinfobound} is
  saturated.
\end{proof}

\begin{cor}
  \label{thm:sicqutrit}
  Upper bounds in Eqs.~\eqref{eq:sicaccinfobound} and
  \eqref{eq:sicinfopowerbound} are tight in dimension $d=3$,
  namely there exist a $3$-dimensional SIC ensemble
  $\set{E}$ and a $3$-dimensional SIC POVM $\set{P}$ such that
  \begin{align*}
    A(\set{E}) = W(\set{P}) = \log\frac32.
  \end{align*}
  The accessible information is achieved by an orthonormal
  POVM and the informational power by an orthonormal
  ensemble.
\end{cor}

\begin{proof}
  The POVM $\set{P} = \{ \Pi_y \}_{y=1}^9$ with $\Pi_y = \frac13
  \ket{\pi_y}\bra{\pi_y}$ given by
  \begin{align*}
    \ket{\pi_1} &= \ket{0},\\ \ket{\pi_2} &= \frac12 \ket{0}
    + i \frac{\sqrt3}2 \ket{1},\\ \ket{\pi_3} &= \frac12
    \ket{0} - i \frac{\sqrt3}2 \ket{1},
  \end{align*}
  \begin{align*}
    \ket{\pi_4} &= \frac12 \ket{0} + \frac12 \ket{1} +
    \frac1{\sqrt{2}}\ket{2},\\ \ket{\pi_5} &= \frac12
    \ket{0} + \frac12 \ket{1} + e^{\frac23 \pi i}
    \frac{1}{\sqrt{2}}\ket{2},\\ \ket{\pi_6} &= \frac12
    \ket{0} + \frac12 \ket{1} + e^{-\frac23 \pi i}
    \frac{1}{\sqrt{2}}\ket{2},
  \end{align*}
  \begin{align*}
    \ket{\pi_7} &= \frac12 \ket{0} - \frac12 \ket{1} +
    \frac1{\sqrt{2}}\ket{2},\\ \ket{\pi_8} &= \frac12
    \ket{0} - \frac12 \ket{1} + e^{\frac23 \pi i}
    \frac{1}{\sqrt{2}}\ket{2},\\\ket{\pi_9} &= \frac12
    \ket{0} - \frac12 \ket{1} + e^{-\frac23 \pi i}
    \frac{1}{\sqrt{2}}\ket{2},
  \end{align*}
  is a SIC POVM in dimension $3$. Notice that states
  $\ket{\pi_i}$ with $i=1,2,3$ lie on a plane on the Bloch
  sphere, and the same for $i=4,5,6$ and $i=7,8,9$.

  The ensemble $\set{E} = \{ \rho_x \}_{x=1}^3$ with $\rho_x
  = \frac13 \ket{\psi_x}\bra{\psi_x}$ given by
  \begin{align*}
    \ket{\psi_1} &= \ket{2},\\ \ket{\psi_2} &= -
    \frac1{\sqrt{2}} \ket{0} + \frac1{\sqrt{2}}
    \ket{1},\\ \ket{\psi_3} &= \frac1{\sqrt{2}} \ket{0} +
    \frac1{\sqrt{2}} \ket{1},
  \end{align*}
  saturates the upper bound in
  Eq.~\eqref{eq:sicinfopowerbound}.  Notice that
  $\Tr[\rho_1\Pi_y] = 0$ if $y=1,2,3$ and $\Tr[\rho_1\Pi_y]
  = 1/18$ otherwise; $\Tr[\rho_2\Pi_y] = 0$ if $y=4,5,6$ and
  $\Tr[\rho_2\Pi_y] = 1/18$ otherwise; $\Tr[\rho_3\Pi_y] =
  0$ if $y=7,8,9$ and $\Tr[\rho_3\Pi_y] = 1/18$
  otherwise. Notice also that $\set{E}$ is an orthonormal
  ensemble. Since the average state of the ensemble
  $\set{E}$ is $\frac1d \openone$, due to
  Lemmas~\ref{lmm:duality} and~\ref{lmm:average} also the
  upper bound in Eq.~\eqref{eq:sicaccinfobound} is
  saturated.
\end{proof}

Notice how the optimal construction in the case of
Corollary~\ref{thm:sicqutrit} is given by orthonormal
ensemble and POVM, contrarily to the case of
Corollary~\ref{thm:sicqubit}, in which the optimal
construction still involves the tetrahedral
configuration. An analogous observation was very recently
reported also by Szymusiak in Ref.~\cite{Szy14}, where a
characterization of the maximally informative ensembles for
group-covariant three-dimensional SIC POVMs is provided.

We can now compare the bounds provided in
Theorem~\ref{thm:sicbound} with those provided in
Theorem~\ref{thm:povmbound}. A plot of these bounds as a
function of the dimension is provided in
Fig.~\ref{fig:bound}.
\begin{figure}[htb]
  \includegraphics[width=\columnwidth]{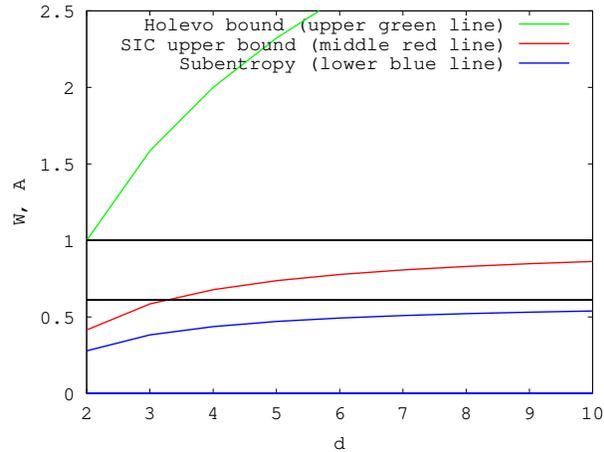}
  \caption{(Color online) Bounds on the accessible
    information and the informational power as given by
    Theorems~\ref{thm:povmbound} and \ref{thm:sicbound}. The
    upper green line is the upper bound for arbitrary
    ensembles and POVMs, namely $\log d$. The middle red
    line is the upper bound for SIC ensembles and SIC POVMs,
    namely $\log\frac{2d}{d+1}$. The lower blue line is the
    lower bound for rank-one POVMs, namely $\log(d) -
    \frac1{\ln(2)} \sum_{n=2}^d \frac1n$. The two horizontal
    black lines are the asymptotes for the latter two
    quantities.}
  \label{fig:bound}
\end{figure}
The lower bounds in Eqs.~\eqref{eq:accinfobound} and
\eqref{eq:infopowerbound}, as well as the upper bounds in
Eqs.~\eqref{eq:sicaccinfobound} and
\eqref{eq:sicinfopowerbound}, are asymptotically finite,
namely
\begin{align*}
  \log{d} - \frac1{\ln2} \sum_{i=2}^d \frac1d &
  \longrightarrow \frac{1 - \gamma}{\ln2} \sim
  0.60995,
\end{align*}
where $\gamma$ is the Euler constant, and
\begin{align*}
 \log{\frac{2d}{d+1}} & \longrightarrow 1.
\end{align*}
Then, while there exist ensembles (resp., POVMs) whose
accessible information (resp., informational power) grows as
$\log d$, that of SIC ensembles (resp., SIC POVMs) is
asymptotically upper bounded by one bit.

\section{Conclusion and outlook}\label{sec:conclusion}

The purpose of this work was twofold. First, we provided
tight upper and lower bounds, depending only on the
dimension of the system, on the accessible information and
the informational power of any preparation and
measurement. While both upper and lower bounds for
accessible information can be trivially saturated, we found
that, for the case of informational power, the lower bound
is saturated by the Scrooge measurement. Second, we focused
on the case of SIC preparations and measurements, providing
stronger upper and lower bounds on accessible information
and informational power, and proved their tightness in the
case of qubits and qutrits. As a corollary of our upper
bound, we showed that the accessible information and the
informational power of SIC preparations and measurements are
always upper bounded by one bit. As a corollary of our lower
bound, on the other hand, we provided a strategy for the
extraction of mutual information from SIC ensembles or by
SIC measurements, beating the ``pretty-good'' one.

Our results can be generalized to arbitrary (in general, not
discrete nor continuous) ensembles and POVMs, upon
redefining the mutual information in terms of relative
entropy~\cite{Cov06}. In the most general case, a
POVM~\cite{Hol73} is defined as an operator valued function
$\set{P} = \{ \Pi(\Delta) \}$ from the collection of Borel
subsets on the real line such that $\Pi(\emptyset) = 0$,
$\Pi(\mathbb{R}) = \openone$, and for every sequence of
disjoint Borel subsets $\Delta_y$ one has that $\Pi(\cup_y
\Delta_y) = \sum_y \Pi(\Delta_y)$. It should be noted that
the sets of discrete and that of continuous POVMs are dense
subsets of the set of general POVMs in a suitable weak
topology, so that the present definition of accessible
information is consistent with that for the general
case. The most general definition of ensemble is
analogous. The discussion of the most general case is out of
the scope of the present manuscript.

We notice that no non-trivial upper and lower bounds on the
informational power as functions of the POVM are known
(while analogous bounds for the accessible information as
function of the ensemble are provided by the
Holevo-Yuen-Ozawa bound~\cite{Hol73, YO93} and the
subentropy bound~\cite{JRW94}). For any POVM $\set{P} = \{
\Pi_x \}$, the subentropy of the ensemble $\set{E} = \{
\rho^{1/2} \Pi_x \rho^{1/2}\}$, for any normalized state
$\rho$, is a lower bound on the informational power of
$\set{P}$, as it follows immediately from
Lemma~\ref{lmm:duality} (see also Refs.~\cite{DDS11,
  OCMB11}). Unfortunately, a general method to maximize over
$\rho$, thus making such bound tighter, is not known.

We conclude our work by adding a few items to the long list
of conjectures related to SIC ensembles and SIC POVMs. We
start by noticing that in dimension $d \ge 4$ it is not
clear whether the upper bounds in Theorem~\ref{thm:sicbound}
are tight or not. Nonetheless, some light can be shed by
numerical optimization. Indeed, from the proof of
Theorem~\ref{thm:sicbound} it follows that necessary
condition for the upper bound in
Eq.~\eqref{eq:sicinfopowerbound} to be tight is that the
bound $H(Y|X{=}x) \ge \log(d(d+1)/2)$ is tight. Then, for
the $16$ $Z_4 \times Z_4$ covariant $4$-dimensional SIC
POVMs listed in Ref.~\cite{RBSC04}, we numerically searched
for the state $\rho_x$ minimizing the entropy
$H(Y|X{=}x)$. We employed a steepest descent algorithm
similar to that proposed in Ref.~\cite{DDS11} for several
($> 1000$) seed states chosen with uniform distribution. For
any considered POVMs, the optimal state we found has
$H(Y|X=x) \sim 3.433$, so the bound $H(Y|X=x) \ge \log10
\sim 3.322$ is not saturated. If a state saturating such
bound does not exist, then the upper bound in
Eq.~\eqref{eq:sicaccinfobound} is not tight for $Z_4 \times
Z_4$ $4$-dimensional SIC POVM.

This rises another intriguing question: is the informational
power the same for any SIC POVM in fixed dimension? And
also: for any SIC POVM in arbitrary dimension, does there
exist a maximally informative ensemble with average state
$d^{-1} \openone$? Of course, the answer to this question
would be affirmative if all SIC POVMs were proved to be
covariant under the action of some discrete group. Moreover,
if true, due to Lemmas~\ref{eq:duality}
and~\ref{lmm:average}, this would also imply an affirmative
answer to the following question: is the informational power
of any SIC POVM $\set{P} = \{ \Pi_x \}$ equal to the
accessible information of the corresponding SIC ensemble
$\set{E} = \{ d^{-1} \Pi_x \}$?  We believe that these
tantalizing open problems well deserve future investigation.

\section*{Acknowledgments}

The authors are grateful to Giacomo Mauro D'Ariano, Paolo
Perinotti, and Massimiliano F. Sacchi for very useful
discussions, comments, and suggestions. This work was
supported by JSPS (Japan Society for the Promotion of
Science) Grant-in-Aid for JSPS Fellows No. 24-0219. M.O. and
F.B. were supported by the JSPS KAKENHI,
No. 26247016. F. B. was supported by the Program for
Improvement of Research Environment for Young Researchers
from Special Coordination Funds for Promoting Science and
Technology (SCF) commissioned by the Ministry of Education,
Culture, Sports, Science and Technology (MEXT) of Japan.

\end{document}